\documentclass[a4paper,12pt]{article}   % MikTeX
\usepackage[utf8]{inputenc} % utf8 cp866 cp1251
\usepackage{hyperref}
\usepackage[english]{babel}
\usepackage{ifthen}
\usepackage{amssymb,amsmath,amsthm}

 \newif\ifNoRemark
    \def\addtheorem#1#2#3#4{ % \usepackage{ifthen} needed
    \ifthenelse{\expandafter\isundefined\csname the#2\endcsname}{\newcounter{#2}}{}
    \newenvironment{#1}[1][\global\NoRemarktrue]% No Remark by default
     {\par\addvspace{2mm}\noindent
       \refstepcounter{#2}{\bf #3~\csname the#2\endcsname
      \vphantom{##1}\ifNoRemark.\ \else\ (##1).\fi}\begingroup #4}%
     {\endgroup\par\addvspace{1mm}\global\NoRemarkfalse}
    \expandafter\newcommand\csname b#1\endcsname{\begin{#1}}
    \expandafter\newcommand\csname e#1\endcsname{\end{#1}}
    }

\newtheorem{theorem}{Theorem}
\newtheorem{lemma}{Lemma}
\newtheorem{proposition}{Proposition}
\newtheorem{corollary}{Corollary}

\begin{document}

\title{Upper bounds on the numbers of binary plateaued and bent
functions }

\author{ V. N. Potapov\\
 Sobolev Institute of Mathematics\\
{ vpotapov@math.nsc.ru}}

\maketitle

\begin{abstract}
We prove that the logarithm of the number of binary $n$-variable
bent functions is asymptotically  less than $\frac{11}{32}2^n$ as
$n\rightarrow\infty$. We also prove an asymptotic upper bound on the
number of $s$-plateaued functions.


 MSC[2020] 94D10, 94A60,  06E30
\end{abstract}

\section{Introduction}

Bent functions are maximally nonlinear Boolean functions with an
even number of variables and are optimal combinatorial objects. In
cryptography, bent functions are used in  block ciphers. Moreover,
bent functions have many theoretical applications  in discrete
mathematics. Full classification of bent functions would be
 useful for combinatorics and cryptography. But constructive
classifications and enumerations of bent functions in $n$ variables
are likely impossible for large $n$. The numbers of $n$-variable
bent functions are only known for $n\leq 8$. There exist  $8$ bent
functions  for $n=2$, $896$ for $n=4$, approximately $2^{32.3}$ for
$n=6$ and  $2^{106.3}$ for $n=8$ \cite{LL}. Thus, lower and upper
asymptotic bounds of the number of bent functions are very
interesting (see \cite[Chapter 4.4]{Mes0}, \cite[Chapter 13]{Tok0}).

Currently, there exists a drastic gap between the upper and lower
bounds on the number of bent functions. Let
$\mathcal{N}(n)=\log_2|\mathcal{B}(n)|$, where $\mathcal{B}(n)$ is
the set  of Boolean bent functions in $n$ variables.
 New   asymptotic  lower
bounds  on the number of bent functions is recently proven in the
binary case \cite{PTT}  and  in the case of  finite fields of odd
characteristic \cite{PF}. In the binary case it is
$\mathcal{N}(n)\geq \frac{3n}{4}2^{n/2}(1+o(1))$ as $n$ is even and
$n\rightarrow\infty$.   This bound is slightly better than the bound
$\mathcal{N}(n)\geq \frac{n}{2}2^{n/2}(1+o(1))$ based on the
Maiorana--McFarland construction of bent functions. It is well known
(see e.g.\,\cite{Carlet}, \cite{CarMes}, \cite{Mes0}) that the
algebraic degree of a binary bent function in $n$ variables is not
greater than $n/2$. Therefore, $\mathcal{N}(n)\leq
{\sum\limits_{i=0}^{n/2} {n \choose i}}= 2^{n-1}+ \frac12{n \choose
n/2} $. There are some nontrivial upper bounds of
$|\mathcal{B}(n)|$. But, bounds from \cite{CK} and \cite{Ag1} are of
the same type $\mathcal{N}(n)\leq 2^{n-1}(1+o(1))$ up to the
asymptotic of logarithm.
 An
 upper bound $\mathcal{N}(n)\leq \frac{3}{4}\cdot2^{n-1}(1+o(1))$
 is proven in \cite{P21}. In this paper we improved latter
  bound and obtained that $\mathcal{N}(n)< \frac{11}{16}\cdot2^{n-1}(1+o(1))$ (Theorem \ref{number_bent}).
 Note that
Tokareva's conjecture (see \cite{Tok} and \cite{Mes0}) on the
decomposition of Boolean functions into sums of bent functions
implies that $\mathcal{N}(n)\geq \frac{1}{2}2^{n-1}+\frac14{n
\choose n/2}$.

The new bound mentioned above  is asymptotic. One can use the
proposed method to find a non-asymptotic upper bound on the number
of bent functions. But for the fixed $n=6$ and $n=8$ such bound is
greater than $\frac{11}{32}\cdot2^{n}$ approximately twice. The main
reason of this difference lies in the cardinality of the middle
layer of the $n$-dimensional Boolean cube. This cardinality is
asymptotically negligible, but that is not in the case for $n=6$ and
$n=8$.

The new upper bound on the number of bent functions is based on new
asymptotic  upper bound on the number of $s$-plateaued Boolean
functions in $n$ variables. $s$-Plateaued functions are a
generalization of bent functions which are the same as $0$-plateaued
functions. Plateaued functions can combine important cryptographic
properties of nonlinearity and correlation immunity (see
e.g.\,\cite{Tar}). Let $\mathcal{N}(n,s)$ be the logarithm of the
number of such functions.  In Theorem \ref{number_plat} (a) we prove
that
$$\mathcal{N}(n,s)\leq
\left(b\left(n-2,\left\lceil\frac{n-s}{2}\right\rceil+1\right)\left(1+\frac{3}{8}\log
6\right)+2^{n-2}\left(\hbar\left(\frac{1}{2^s}\right)+\frac{1}{2^s}\right)\right)(1+o(1))$$
as $n\rightarrow\infty$, where $\hbar$ is  Shannon's entropy
function and $b(n,r)$ is the cardinality of balls  with radius $r$
in the $n$-dimensional Boolean cube. This  bound is not tight but
sufficient to disprove the following conjecture on derivatives of
bent functions. Tokareva (see  \cite{Tok16, Shap}) conjectured that
each balanced Boolean function $f$ in an even number of variables
$n$ of algebraic degree at most $n/2 -1$  is a derivative of some
bent function if $f (x) = f (x \oplus y)$ for every vector $x$ and
some nonzero vector $y$. It is true for $n\leq6$
\cite[Theorem~4]{Shap} but based on  Theorem \ref{number_plat} (b)
it is proved  that this conjecture is false when $n$ is large enough
(see \cite{Pot23}).

 The method of proving the above bounds implies a
storage algorithm for bent and plateaued functions. Essentially we
calculate the number of bits needed to define all the values of the
function.
 A quantity of
bits required by the algorithm  is equal to the corresponding upper
bound $\mathcal{N}(n,s)$. In practice we need bent and plateaued
functions with some special properties. See, for example, some
recent constructions from \cite{JL}, \cite{Li} and \cite{Zhang}.
Therefore, methods of compact storage of $n$-variable bent and
plateaued functions may be useful for large $n$.

\section{Fourier transform}

Let $\mathbb{F}= \{0, 1\}$. The set $\mathbb{F}^n$ is called a
Boolean hypercube (or a Boolean $n$-cube). $\mathbb{F}^n$ equipped
with  coordinate-wise modulo 2 addition $\oplus$ can be considered
as an $n$-dimensional vector space. Functions
 $\phi_x(y)=(-1)^{\langle x,y\rangle}$ are called characters. Here $\langle
x,y\rangle=x_1y_1\oplus\dots\oplus x_ny_n$ is the inner product. Let
$G$ be a function that  maps from the Boolean hypercube to real
numbers. The Fourier transform of $G$ is defined by the formula
$\widehat{G}(y)=(G,\phi_y)=\sum\limits _{x\in
\mathbb{F}^n}G(x)(-1)^{ \langle x,y\rangle}$, i.e., $\widehat{G}(y)$
are the coefficients of the expansion of $G$ with respect to the
basis of characters. We can define the Walsh--Hadamard transform of
a Boolean function $f:\mathbb{F}^n\rightarrow \mathbb{F}$  by the
formula $W_f(y)=\widehat{(-1)^f}(y)$, i.e.,
$$W_f(y)=\sum\limits_{x\in \mathbb{F}^n}(-1)^{f(x)\oplus \langle
x,y\rangle}.$$ A Boolean function $b$ is called a bent function if
$W_b(y)=\pm 2^{n/2}$ for all $y\in \mathbb{F}^n$.  So,
$W_b=2^{n/2}(-1)^g$ for some Boolean function $g$. It is well known
(see  e.g. \cite{Carlet, Mes0}) that $g$ is also a bent function and
$W_g=2^{n/2}(-1)^f$. Such bent functions $b$ and $g$ are called
{dual}.
 It is easy to see that $n$-variable bent
functions  exist only if $n$ is even. A Boolean function $p$ is
called an $s$-plateaued function if $W_p(y)=\pm 2^{(n+s)/2}$  or
$W_p(y)=0$ for all $y\in \mathbb{F}^n$. So, bent functions are
$0$-plateaued functions. $1$-Plateaued functions are called
near-bent.

From Parseval's identity  $$\sum\limits_{y\in \mathbb{F}^n}
\widehat{H}^2(y)= 2^{n}\sum\limits_{x\in \mathbb{F}^n}H^2(x),$$
where $H:\mathbb{F}^n\rightarrow \mathbb{R}$, it follows
straightforwardly:
\begin{proposition}\label{bentpla00}
For every  $s$-plateaued function,  a part of nonzero values of its
Walsh--Hadamard transform is equal to $\frac{1}{2^s}$.
\end{proposition}

It is well known (see e.g. \cite{Carlet},\cite{Tsf}) that for any
function $H,G:\mathbb{F}^n\rightarrow \mathbb{R}$ it holds
$$\widehat{H*G}= {\widehat{H}\cdot\widehat{ G}} \qquad {\rm
and}\qquad \widehat{(\widehat{H})}=2^nH,$$ where
$H*G(z)=\sum\limits_{x\in \mathbb{F}^n}H(x)G(z\oplus x)$ is a
convolution. Consequently, it holds
\begin{equation}\label{equpperbent11}
2^{n}H*G= \widehat{\widehat{H}\cdot\widehat{ G}} \ \mbox{and}\
\widehat{H}*\widehat{ G}= 2^{n}\widehat{H\cdot G}.
\end{equation}
 Let $\Gamma$ be a subspace of the hypercube.
 %Here and below faces are axes-aligned planes of the hypercube.
 Denote by
   $\Gamma^\perp$  a
dual subspace, i.e., $\Gamma^\perp=\{y\in \mathbb{F}^n : \forall
x\in \Gamma, \langle x,y\rangle=0 \}.$ Let ${\bf1}_S$ be an
indicator function for $S\subset\mathbb{F}^n $. It is easy to see
that for every subspace $\Gamma$ it holds
$\widehat{{\bf1}_{\Gamma^\perp}}=2^{n-\mathrm{dim}\,\Gamma}{\bf1}_{\Gamma}$.
By (\ref{equpperbent11}) we have
\begin{equation}\label{equpperbent1}
{H}*{\bf1}_{\Gamma^\perp}=2^{-\mathrm{dim}\,\Gamma}\widehat{\widehat{H}\cdot\mathbf{1}_{\Gamma}}
\end{equation}
for any subspace  $\Gamma\subset \mathbb{F}^n$. When we substitute
vector $a\in \mathbb{F}^n$ in (\ref{equpperbent1}) we obtain
\begin{equation}\label{equpperbent10}
\sum\limits_{x\in a\oplus\Gamma^\perp}H(x)
=2^{-\mathrm{dim}\,\Gamma}{\sum\limits_{y\in
\Gamma}\widehat{H}(y)(-1)^{y\oplus a}}.
\end{equation}

Denote by $\mathrm{supp}(G)=\{x\in \mathbb{F}^n : G(x)\neq 0\}$ the
support of $G$. Without any confusion, we will consider the support
of a real-valued function as a Boolean function. We need the
following known property of bent functions (see e.g.\,\cite{Mes0}).

\begin{proposition}\label{bentpla1}
Let $f$ be an $n$-variable bent function and let $\Gamma$ be a
hyperplane obtained by fixing one coordinate to $0$. Consider
$h={f}\cdot{\bf1}_{\Gamma}$ as an $(n-1)$-variable function. Then
$h$ is a  $1$-plateaued function.
\end{proposition}
\begin{proof}
By the definition we have $2^{n/2}(-1)^f= \widehat{(-1)^g}$,  where
a bent function $g$ is  dual of $f$. By (\ref{equpperbent1})
$$2^{\frac{n}{2}-1}(-1)^g*{\bf1}_{\Gamma^\perp}=\widehat{(-1)^{f}\cdot{\bf1}_{\Gamma}}.$$

For a nonzero $a\in\Gamma^\perp$ and any $x\in \Gamma$ we obtain
$$(-1)^g*{\mathbf 1}_{\Gamma^\perp}(x)=(-1)^g(x)+(-1)^g(x\oplus a)=\pm 2\ {\rm
or }\ 0.$$ Then
$\widehat{(-1)^{f}\cdot{\bf1}_{\Gamma}}(x)=\pm2^{\frac{(n-1)+1}{2}}$
or $0$.  It is easy to see that
$\widehat{(-1)^h}(x)=\widehat{(-1)^{f}\cdot{\bf1}_{\Gamma}}(x)=\widehat{(-1)^{f}\cdot{\bf1}_{\Gamma}}(x\oplus
a)$. Consequently, $h$ is a $1$-plateaued function by the
definition.
\end{proof}

\begin{proposition}\label{cor:Sar}
Suppose that $f$ and $g$ are Boolean functions in $n$ variables. For
any subspace $\Gamma\subset \mathbb{F}^n$ if
$W_f|_\Gamma=W_g|_\Gamma$ then $\sum\limits_{x\in
z\oplus\Gamma^\perp}(-1)^{f(x)}=\sum\limits_{x\in
z\oplus\Gamma^\perp}(-1)^{g(x)}$ for any $z\in \mathbb{F}^n$.
\end{proposition}
\begin{proof} It follows from (\ref{equpperbent1}). Indeed it holds
$\widehat{(-1)^f}\cdot{\bf1}_{\Gamma}=\widehat{(-1)^g}\cdot{\bf1}_{\Gamma}$
by the conditions of the lemma. Then
$(-1)^f*{\bf1}_{\Gamma^\perp}=(-1)^g*{\bf1}_{\Gamma^\perp}$. It is
clear that $((-1)^f*{\bf1}_{\Gamma^\perp})(z)=\sum\limits_{x\in
z\oplus\Gamma^\perp}(-1)^{f(x)}$. This completes the
proof.\end{proof}

\section{M\"obius  transform}

Denote by $\mathrm{wt}(z)$  the number of units in $z\in
\mathbb{F}^n$. Every Boolean function $f$ can be represented in the
algebraic normal form:
\begin{equation}\label{equpperbent00}
 f(x_1,\dots,x_n)=\bigoplus\limits_{y\in
    \mathbb{F}^n}M[f](y)x_1^{y_1}\cdots x_n^{y_n},
\end{equation}
     where $x^0=1, x^1=x$, and
$M[f]:\mathbb{F}^n\rightarrow \mathbb{F}$ is the M\"obius transform
of $f$. It is well known that
\begin{equation}\label{equpperbent0}
M[f](y)=\bigoplus\limits_{x\in \Gamma_y}f(x)
\end{equation}
where $\Gamma_y=\{(x_1,\dots,x_n)\in \mathbb{F}^n : x_i=0 \
\mbox{\rm if} \ y_i=0\}$ is a subspace of $\mathbb{F}^n$. Note that
$M[M[f]]=f$ for each Boolean function (see \cite[Theorem
1]{Carlet}).
 The degree of this polynomial   is called  the algebraic
degree  of $f$.

Denote by $b(n,r)$ the cardinality of a ball $B_{n,r}$ with radius
$r$ in $\mathbb{F}^n$, i.e., $b(n,r)=|\{x\in \mathbb{F}^n :
\mathrm{wt}(x)\leq r\}|$.
 By properties of the M\"obius transform, the number
of $n$-variable Boolean functions $f$ such that $\mathrm{
deg}\,f\leq r$ is equal to $2^{b(n,r)}$.

%equal to $2^{\sum\limits_{i=0}^d {n \choose i}}$.

\begin{lemma}\label{ball}
Suppose that $f$ and $g$ are $n$-variable Boolean functions  and\\
$\max\{{\rm deg}(f),{\rm deg}(g)\}\leq r$. If
$f|_{B_{n,r}}=g|_{B_{n,r}}$ then $f=g$.
\end{lemma}
\begin{proof}
 By the hypothesis  of the
lemma and (\ref{equpperbent00}), we have
    $M[f](y)=M[g](y)=0$ if $\mathrm{wt}(y)>r$. By (\ref{equpperbent0})
     for any $y\in \mathbb{F}^n$ such that $\mathrm{wt}(y)=r+1$,  we obtain
$$M[f](y)=\bigoplus\limits_{x\in \Gamma_y}f(x)=
f(y)\oplus\bigoplus\limits_{x\in \Gamma_y\cap B_{n,r}}f(x) $$
$$=f(y)\oplus\bigoplus\limits_{x\in \Gamma_y\cap
B_{n,r}}g(x)= f(y)\oplus M[g](y)\oplus g(y).$$ Therefore, $f(y)=
g(y)$ for any $y\in B_{n,r+1}$. By induction on weights
$\mathrm{wt}(y)$, we obtain that $f(y)= g(y)$ for all $y\in
\mathbb{F}^n$.
\end{proof}

\begin{lemma}[\cite{Carlet}, Theorem 2]\label{lemdeg} Let $f$ be an $n$-variable Boolean function. Suppose for every
$y\in \mathbb{F}^n$ it holds $\widehat{(-1)^f}(y)=2^km(y)$, where
$m(v)$ is integer. Then ${\rm deg}(f) \leq n-k+1$. \end{lemma}

\begin{corollary}[\cite{Carlet}, Proposition 96]\label{cordeg}
The algebraic degree of $n$-variable $s$-plateaued functions is not
greater than $\frac{n-s}{2}+1$.
\end{corollary}

Note that algebraic degrees of bent ($0$-plateaued) functions is
$n/2$ at most (see e.g.\,\cite{Carlet}, \cite{CarMes}, \cite{Mes0}),
but for $1$-plateaued functions the upper bound $\frac{n+1}{2}$ is
sharp.

\begin{proposition}\label{bentpla19}
Let $f$ be an $n$-variable bent function. Then for any hyperplane
$\Gamma$ the algebraic degree of the Boolean function $h=\mathrm{
supp} (\widehat{(-1)^{f}\cdot{\bf1}_{\Gamma}})$ is not greater than
$n/2$.
\end{proposition}
\begin{proof}
 By  (\ref{equpperbent1}) we obtain that $h=\mathrm{
 supp}((-1)^g*{\bf1}_{\Gamma^\perp})$, where a bent
 function  $g$ is dual of $f$. Let $\Gamma^\perp=\{\bar{ 0}, a\}$. Then
 $(-1)^g*{\bf1}_{\Gamma^\perp}(x)=(-1)^{g(x)}+(-1)^{g(x\oplus a)}$.
 Consequently, $h(x)=g(x)\oplus g(x\oplus a)\oplus 1$. Thus,
 $\mathrm{deg}\,h\leq\mathrm{deg}\,g\leq\frac{n}{2}$.
 \end{proof}

\section{Subspace distribution}

We will use the following well-known criterium (see,
e.g.\,\cite[Proposition 96]{Carlet}) which is also true in the
nonbinary case (\cite[Theorem 2]{MOA}).
\begin{lemma}\label{lem:conv}
An $n$-variable Boolean function $f$ is  $s$-plateaued  if and only
if it holds\\ ${(-1)^f}*{(-1)^f}*{(-1)^f}=2^{n+s}{(-1)^f}$.
\end{lemma}

Consider an $n$-variable $s$-plateaued Boolean function $f$ and any
fixed $x\in \mathbb{F}^n$. There are $V={n\brack
2}_2=\frac{(2^n-1)(2^n-2)}{6}$ $2$-dimensional affine subspaces such
that any of them contains $x$. Let $S(x)$ be the number of subspaces
containing an odd number of zero values of $f$. By Lemma
\ref{lem:conv} we obtain

\begin{proposition}\label{distrib0}
For any fixed $x\in \mathbb{F}^n$, it holds $\frac{S(x)}{V}=\frac12
-\frac12\cdot \frac{2^{n+s}-3\cdot2^n+2}{(2^n-1)(2^n-2)}$.
\end{proposition}
\begin{proof}
We can rewrite the formula from Lemma \ref{lem:conv} by the
following form
$$2^{n+s}(-1)^{f(x)}=\sum\limits_{z\in\mathbb{F}^n}\sum\limits_{y\in\mathbb{F}^n}(-1)^{f(y)\oplus f(y\oplus z)\oplus f(x\oplus z)}=\sum\limits_{z\in\mathbb{F}^n}\sum\limits_{y\in\mathbb{F}^n}(-1)^{f(x\oplus y)\oplus f(x\oplus y\oplus z)\oplus f(x\oplus z)}.$$

It is easy to see that if two elements of $\{x,x\oplus y, x\oplus z,
x\oplus y\oplus z\}$ are coincide then the two remaining elements
are also coincide. In this case $(-1)^{f(x)\oplus f(x\oplus y)\oplus
f(x\oplus z)\oplus f(x\oplus y\oplus z)}=1$. Let $U$ be the set of
such couples $\{y,z\}$ that $\{x,x\oplus y, x\oplus z, x\oplus
y\oplus z\}$ is a $2$-dimensional affine subspace.
 By the
inclusion-exclusion formula, we obtain that
$$\sum\limits_{\{y,z\} \in U}(-1)^{f(x)\oplus f(x\oplus
y)\oplus f(x\oplus z)\oplus f(x\oplus y\oplus z)}=2^{n+s}-3\cdot
2^n+2.$$ It is easy to see that every subspace $\{x,x\oplus y,
x\oplus z, x\oplus y\oplus z\}$ occurs $6$ times in the sum above.
Consequently, it holds equation $$\frac{2^{n+s}-3\cdot
2^n+2}{6}=(V-S(x))-S(x).$$ The extraction of $\frac{S(x)}{V}$ from
the last equation  completes the proof.
\end{proof}

Thus we have two equations: $\frac{S(x)}{V}=\frac12
+\frac{1}{2(2^{n-1}-1)}$ for every bent function and
$\frac{S(x)}{V}=\frac12 +\frac{1}{2(2^{n}-1)}$ for every
$1$-plateaued function. Note that for  bent functions $f$, $f(\bar{
0})=0$,  numbers of  linear subspaces such that contain $1$, $2$,
$3$ or $4$  zero values of $f$ do not depend on  $f$ (see
\cite{PA}).
%Some new results on value distributions of nonbinary
%bent functions are in \cite{Kolsch}.

We will use the following property of bent and plateaued functions.

\begin{proposition}[\cite{Carlet}, \cite{CarMes}, \cite{Mes0}]\label{splat12}
Let $f:{\mathbb{F}}^n\rightarrow {\mathbb{F}}$ be an $s$-plateaued
function, let $A:{\mathbb{F}}^n\rightarrow {\mathbb{F}}^n$ be a
non-degenerate affine transformation and let
$\ell:{\mathbb{F}}^n\rightarrow {\mathbb{F}}$ be an affine function.
Then $g=(f\circ A)\oplus\ell$ is an $s$-plateaued  function.
\end{proposition}

The functions $f$ and $g$ satisfied the conditions of Proposition
\ref{splat12} are called EA-equivalent. It is easy to see that the
cardinality of any equivalence class is not greater than
$a_n=2^{n^2+n+1}(1+o(1))$. Note that two EA-equivalent functions $f$
and $g$ have the same algebraic degree as ${\rm deg}(f)>1$.

There are eight  $2$-variable Boolean functions such that take value
$0$ even times.  All of them are affine. Six of them take value $0$
two times and the other take value $0$ four or zero times. Consider
a $2$-dimensional affine subspace $\Gamma$ and an $n$-variable
Boolean function $g$. Let $g$ take value $0$ even times on $\Gamma$.
It is easy to see that $3/4$ among functions of the  set $\{g\oplus
\ell : \ell\ \mbox{\rm is an affine function}\}$  take value $0$ two
times and the other take value $0$ four or zero times. Consequently,
from Propositions \ref{distrib0} and  \ref{splat12} we deduced:

\begin{corollary}\label{distrib00}

Let $\Gamma$ be a $2$-dimensional face, i.e, axes-aligned plane
which can be obtained by fixed all with the exception of two
coordinates, and let $f:{\mathbb{F}}^n\rightarrow {\mathbb{F}}$ be
an $s$-plateaued function. There exists   a non-degenerate affine
transformation $A$ and an affine function $\ell$ such that the
$s$-plateaued  function $g=(f\circ A)\oplus\ell$ satisfies the
following conditions.

{\rm (a)} The part of faces  $\Gamma\oplus y$, $y\in
{\mathbb{F}}^n$, that contain an odd number of zero values of $g$,
is less than $\frac12$ if $s>1$ and less than
$\frac12+\frac{1}{2^n}$ if $s=1$.

{\rm (b)} Among the faces  $\Gamma\oplus y$, $y\in
B_{n,r}\subset{\mathbb{F}}^n$,
 that contain an even number of zero values of $g$, not less than
one fourth part contain four or zero values $0$.
\end{corollary}
\begin{proof}
Firstly, we can find $A$ to provide condition (a). Let $s>1$.
Suppose that  the fraction of faces $A^{-1}(\Gamma\oplus y)$, $y\in
{\mathbb{F}}^n$,  containing  an odd number of zero values of $f$,
is not less than $\frac12$ for every non-degenerate affine
transformation $A$. Then  at least  half of $2$-dimensional affine
subspaces contain odd numbers of zero values of $f$. It is
contradict to Proposition \ref{distrib0}. Therefore, we can fixed a
non-degenerate affine transformation $A$ such that $g= f\circ A$
satisfies condition (a). The case $s=1$ is similar.

 Secondly, we can find $\ell$ to satisfy
condition (b). Indeed, we can choose $\ell$ to provide (b), since as
mentioned above, this distribution is on the average for all $\ell$.
By adding any affine function we save the parity of the number of
zero values of $g$ on every $2$-dimensional affine subspace. So, we
preserve condition (a). Consider the distribution of even numbers of
zero values of $g$ on the faces  $\Gamma\oplus y$, $y\in
{\mathbb{F}}^n$. It is easy to see that the average distribution
over balls with fixed radius $r$ and  centers $y\in {\mathbb{F}}^n$
is equal to the distribution over the Boolean hypercube. Then there
exists a ball with center $e\in {\mathbb{F}}^n$ such that $g=
(f\circ A)\oplus\ell$ has the same or better distribution on the
ball with center $e$. Then we can exchange $A$ to $A\oplus e$ to
provide the required distribution on $B_{n,r}$.
\end{proof}

Note that a random Boolean function has the required distribution of
zero values in a $2$-dimensional face, i.e., zero or four $0$s with
probability $\frac18$, one or three $0$s with probability $\frac12$,
two $0$s with probability $\frac38$.

Let $p_0$ be a probability of an even number of zero values in a
$2$-dimensional face and let $p_1$
 be a probability of an odd number of zero values in a
$2$-dimensional face. Moreover,  $p'_0$ is the  probability of two
zero values in a $2$-dimensional face and $p'_0\leq 3p_0/4$. How
many bits on average we need to count four values $(-1)^{g(x)}$ in a
$2$-dimensional face  $\Gamma\oplus y$ from their sum? We use the
following simple proposition.

\begin{proposition}\label{dist33}
Let $M$ be a set of words with length $n$ and let $k_i>0$ be a
number of different symbols in $i$th place in every word from $M$.
If $p_m=|\{i\in \{1,\dots,n\} : k_i=m\}|/n$ then
$$\frac{\log_2|M|}{n}= \sum_mp_m\log_2m.$$
\end{proposition}
\begin{proof}
By the definition, $|M|=\prod^n_{i=1}k_i$. Rearranging the factors,
we obtain that $|M|=\prod_m m^{p_mn}$. Taking the logarithm of both
sides of the previous equality  we deduce the required equality.
\end{proof}

Consequently, to find all  values of function in a $2$-dimensional
face in the case of two zero values we need $\log_2 6$ bits, in the
case of odd zero value we need $2$ bits, in the case of $0$ or $4$
zero values we do not need extra bits.
 Therefore,  under conditions (a) and (b) from  Corollary
\ref{distrib00}, it is sufficient  $p'_0\log_2 6+2p_1\leq
1+\frac{3}{8}\log_2 6 =\alpha\approx 1.969$  (or
$1+\frac{1}{2^{n-1}}+\frac{3}{8}\log_2 6=\alpha_n$ if $s=1$) bits on
average for finding four values $(-1)^{g(x)}$ in a $2$-dimensional
face $\Gamma\oplus y$ from their sum. It is easy to see that
$\alpha_n\rightarrow\alpha$ as $n\rightarrow\infty$. So, we obtain
the following statement from Corollary \ref{distrib00} and
Proposition \ref{dist33}.

\begin{corollary}\label{distrib001}
For every $n$-variable $s$-plateaued function there exists an
EA-equivalent function $g$ and $2$-dimensional face $\Gamma$ such
that if we know sums of values of $(-1)^g$ on all $\Gamma\oplus y$,
$y\in B_{n,r}$, then it is sufficient $\alpha b(n,r)$ (or $\alpha_n
b(n,r)$ if $s=1$) extra bits to  identify $g$ on $B_{n,r}$.
\end{corollary}

\section{Main results}

In the previous section we proved that in every EA-equivalence class
there exists an $s$-plateaued function $f$ satisfying the conditions
of Corollary \ref{distrib00}. Now we estimate the number of bits
sufficient to determine $f$.
 We will use the following combinatorial version of Shannon's
source coding theorem.

\begin{proposition}[see e.g. \cite{Kri}]
Let $M_n$ be a set of equally composed words with length $n$ over
alphabet $\mathcal{A}$ and let $p_i>0$ be a frequency of $i$th
symbol of $\mathcal{A}$ in every word from $M_n$. Here
$\sum\limits_{i=1}^{|\mathcal{A}|}p_i=1$. Then
$$\frac{\log_2|M_n|}{n}=\sum\limits_{i=1}^{|\mathcal{A}|}p_i\log_2\frac{1}{p_i}+\varepsilon(n,M_n),$$
where  $\sup_{M_n}|\varepsilon(n,M_n)|\rightarrow 0$ as
$n\rightarrow\infty$.
\end{proposition}

The sum $\sum_ip_i\log_2\frac{1}{p_i}$ is called Shannon's entropy
of source with probability $p_i$ of $i$th symbol. Denote by $\hbar$
Shannon's entropy function in the case of two symbols, i.e.,
$\hbar(p)=-p\log p- (1-p)\log (1-p)$ for $p\in (0,1)$.

Let $\mathcal{N}(n,s)$ be the binary logarithm of the number of
$n$-variable $s$-plateaued Boolean functions. The logarithm of the
number  of bent functions we denoted by $\mathcal{N}(n)$. Since the
Walsh--Hadamard transform is a bijection, $\mathcal{N}(n,s)$ is not
greater than the number of bits such that is sufficient to identify
$W_f$ for an $s$-plateaued function $f$. Therefore, by    Shannon's
source coding  theorem and Proposition \ref{bentpla00} we obtain
inequality:
\begin{equation}\label{eqSh}
 \mathcal{N}(n,s)\leq
2^n\left(\hbar\left(\frac{1}{2^s}\right)(1+o(1))+\frac{1}{2^s}\right).
\end{equation}

Let $\mathcal{N}_0(n,1)$ be the binary logarithm of the number of
$n$-variable $1$-plateaued Boolean functions  which are obtained by
the restriction of domain of  $(n+1)$-variable bent functions to
hyperplanes.

\begin{theorem}\label{number_plat}
{\rm (a)}   $\mathcal{N}(n,s)\leq (\alpha
b(n-2,\lceil\frac{n-s}{2}\rceil+1)+2^{n-2}(\hbar(\frac{1}{2^s})+\frac{1}{2^s}))(1+o(1))$
where $\alpha= 1+\frac{3}{8}\log_2 6$, $s>0$ is fixed  and
$n\rightarrow\infty$.

{\rm (b)}  $\mathcal{N}_0(n,1)\leq
b(n-2,\frac{n+1}{2})(\alpha+\frac{3}{2})(1+o(1))$ as
$n\rightarrow\infty$.
\end{theorem}
\begin{proof}
Let $f$ be an $n$-variable $s$-plateaued function. Consider an
$(n-2)$-dimensional face $\Gamma$. Without loss of generality (see
Propositions \ref{bentpla00} and \ref{splat12}) we admit that the
part of nonzero values of $W_f$ in $\Gamma$ is not greater than
$\frac{1}{2^s}$.

Case (a).  For every $s$-plateaued Boolean function $f$ it is
sufficient (by similar way as  in (\ref{eqSh}))
$2^{n-2}(\hbar(\frac{1}{2^s})+\frac{1}{2^s})(1+o(1))$ bits to
identify $W_f\cdot \mathbf{1}_\Gamma$.

Case (b). Suppose that a $1$-plateaued function $f$ is  obtained by
the restriction of domain of a bent function to a hyperplane. By
Proposition \ref{bentpla19}, an algebraic degree of the support $S$
of the Walsh--Hadamard transform of such $1$-plateaued function is
not greater than $\frac{n+1}{2}$. By Lemma \ref{ball}, it is
sufficient to identify $S$ only in a ball $B_{n-2, \frac{n+1}{2}}$.
Then we just need $(\hbar(\frac12) +\frac12)b(n-2,
\frac{n+1}{2})(1+o(1))=\frac{3}{2}b(n-2, \frac{n+1}{2})(1+o(1))$
bits to identify $W_f\cdot \mathbf{1}_\Gamma$ by (\ref{eqSh}).

The last part of the proof is the same for cases (a) and (b).

By (\ref{equpperbent10}), if we know $W_f\cdot \mathbf{1}_\Gamma$
then we can find  sums $\sum\limits_{x\in \Gamma^\perp\oplus a}
(-1)^f(x)$ for any $a\in \mathbb{F}^n$. By Corollary
\ref{distrib00}, we can choose an $s$-plateaued function such that
is  EA-equivalent to $f$ and has  the appropriate distribution of
these sums. By Corollary \ref{cordeg}, an algebraic degree of any
$n$-variable $s$-plateaued Boolean functions $f$ is not greater than
$r=\lceil\frac{n-s}{2}\rceil+1$. Consequently,  by Lemma
\ref{lemdeg} it is sufficient to recognize values of $f$ in a ball
of radius $r$. By Corollary \ref{distrib001}, there exists
$s$-plateaued function $f'$ from the same EA--equivalence class as
$f$ such that $\alpha b(n,r)$ bits is sufficient to recover $f'$ on
$\mathbb{F}^n$. If we know  a EA--equivalence class of the function
then    it is sufficient $\log_2 a_n=n^2+n+1=o(2^n)$ bits to
identify the function.
 Thus, if $W_f\cdot \mathbf{1}_\Gamma$ is given then we need $\alpha
b(n-2,\lceil\frac{n-s}{2}\rceil+1)(1+o(1))$ extra bits to identify
$(-1)^f$.
\end{proof}

\begin{corollary}\label{corplat100}
$\mathcal{N}_0(n,1)< 3.47\cdot2^{n-3}(1+o(1))$ as
$n\rightarrow\infty$.
\end{corollary}

The idea of the new upper bound on the number of the Boolean bent
functions is the  following. We consider a restriction of the domain
of a bent function into a hyperplane. This is a  $1$-plateaued
function and we can evaluate the number of such functions by Theorem
\ref{number_plat} (b). Then we evaluate the number of extra bits
witch we need to recover all values the bent function when we know
it values only on hyperplane. By Lemma \ref{ball} it is sufficient
to identify an $n$-variable bent function only on a ball with radius
$n/2$.

\begin{theorem}\label{number_bent}
$\mathcal{N}(n)\leq \mathcal{N}_0(n-1,1)+2^{n-3}(1+o(1))<
\frac{11}{32}2^n(1+o(1))$ as $n\rightarrow\infty$.
\end{theorem}
\begin{proof}
Let $f$ be an $n$-variable bent function and let $g$ be  dual of $f$
bent function. By Proposition \ref{bentpla1},
$g\cdot\mathbf{1}_\Gamma$ is an $(n-1)$-variable $1$-plateaued
function as  $\Gamma$ is a hyperplane. Presume that
$\Gamma^\perp=\{\bar{ 0}, a\}$.

Now we evaluate a number of extra bits which is sufficient  to
recover $f$ if $g\cdot\mathbf{1}_\Gamma$ is given. By
(\ref{equpperbent10}) we obtain that  sums
$s(x)=(-1)^{f(x)}+(-1)^{f(x+a)}$ are determined by
$g\cdot\mathbf{1}_\Gamma$. Since all derivatives of each bent
function are balanced (see e.g. \cite{Carlet}, Theorem 12), a half
of these sums $s(x)$ are equal to $\pm 2$ and the other half of
these sums are equal to $0$. In the first case we can extract
$(-1)^{f(x)}$ and $(-1)^{f(x+a)}$ from the sum. But in the second
case we need an additional  information to choose $(-1)^{f(x)}=1$
and $(-1)^{f(x+a)}=-1$ or vice versa. Denote by $S_1$ the set of
$x\in \Gamma$ such that $s(x)=\pm 2$.

By Lemma \ref{ball} and Proposition \ref{splat12}, we need to
identify values of $f$ only in some ball with radius $n/2$. It is
easy to see that we can  find a ball $B$ such that $|S_1\cap B\cap
\Gamma|\geq |B\cap \Gamma|/2$. Therefore, it is necessary not
greater than $|B\cap \Gamma|-|S_1\cap B\cap \Gamma| \leq |B\cap
\Gamma|/2=2^{n-3}(1+o(1))$ extra bits to recover $f$ from
$g\cdot\mathbf{1}_\Gamma$. Therefore, we establish that
$$\mathcal{N}(n)\leq \mathcal{N}_0(n-1,1)+2^{n-3}(1+o(1))$$ as
$n\rightarrow\infty$. By Theorem \ref{number_plat}(b), we obtain the
required inequality.
\end{proof}

\section*{Acknowledgements}

The author is grateful to S.~Avgustinovich and S.~Agievich for their
attention to this work and  useful discussions.

The research  has been carried out within the framework of a state
assignment of the Ministry of Education and Science of the Russian
Federation for the Institute of Mathematics of the Siberian Branch
of the Russian Academy of Sciences (project no. FWNF-2022-0017).

\end{document}